\documentclass[a4paper,conference]{IEEEtran}

\IEEEsettopmargin{t}{30mm}
\IEEEquantizetextheight{c}
\IEEEsettextwidth{14mm}{14mm}
\IEEEsetsidemargin{c}{0mm}
\usepackage[utf8]{inputenc}
\usepackage[T1]{fontenc}
\usepackage{url}
\usepackage{ifthen}
\usepackage{cite}
\usepackage[cmex10]{amsmath} %
\usepackage[nolist]{acronym}
\begin{acronym}
\acro{AD}{anchor decoding}
\acro{AGC}[AGC]{automatic gain control}
\acro{API}[API]{application program interface}
\acro{AWGN}[AWGN]{additive white Gaussian noise}
\acro{AVX}[AVX]{advanced vector extensions}

\acro{BCH}{Bose--Chaudhuri--Hocquenghem}
\acro{BDD}{bounded distance decoding}
\acro{BEC}[BEC]{binary erasure channel}
\acro{BEE-PC}{binary message passing based on \ac{EaE} decoding for \acp{PC}}
\acro{BER}[BER]{bit error rate}
\acro{BDMC}[BDMC]{binary \acs{DMC}}
\acro{BI-AWGN}{binary input additive white Gaussian noise}
\acro{BP}[BP]{belief propagation}
\acro{BPSK}{binary phase shift keying} 
\acro{BSC}[BSC]{binary symmetric channel}

\acro{CEL}[CEL]{Communications Engineering Lab}
\acro{CPU}[CPU]{central processing unit}
\acro{DFT}[DFT]{discrete Fourier transform}
\acro{DMC}[DMC]{discrete nemoryless channel}
\acro{DRS}{dynamic reliability score}
\acro{DRSD}{\ac{DRS} decoder}
\acro{dSNR}[dSNR]{design \ac{SNR}}
\acro{DSP}[DSP]{digital signal processing}
\acro{DTP}{decoding transition probability}
\acro{EaE}{error-and-erasure}
\acro{EaED}{\ac{EaE} decoder}
\acro{EEP}{error-evaluator polynomial}
\acro{ELP}{error-locator polynomial}
\acro{EMP}{extrinsic message passing}
\acro{MSP}{minimal stall pattern}
\acro{FEC}[FEC]{forward error correction}
\acro{FER}[FER]{frame error rate}
\acro{FFT}[FFT]{fast Fourier transform}
\acro{GPP}[GPP]{general purpose processor}
\acro{GPC}{generalized \ac{PC}}
\acro{GMD}{generalized minimal distance}
\acro{GMDD}{generalized minimal distance decoder}
\acro{HD}{hard decision}
\acro{HDD}{hard decision decoding}
\acro{HIHO}{hard input hard output}
\acro{HRB}{highly reliable bit}
\acro{iBDD}{iterative \ac{BDD}}
\acro{IDFT}[IDFT]{inverse discrete Fourier transform}
\acro{iEaED}{iterative error-and-erasure decoding}

\acro{KIT}[KIT]{Karlsruhe Institute of Technology}
\acro{LDPC}[LDPC]{low-density parity-check}
\acro{LR}[LR]{likelihood ratio}
\acro{LLR}[LLR]{log-likelihood ratio}
\acro{LTE}[LTE]{Long Term Evolution}
\acro{MAP}[MAP]{maximum a posteriori}
\acro{MC}{miscorrection}
\acro{MD}{miscorrection detection}
\acro{MDS}{maximum distance separable}
\acro{ML}[ML]{maximum likelihood}
\acro{NCG}{net coding gain}
\acro{OSD}{ordered statistics decoding}
\acro{PC}{product code}
\acro{PCM}{parity check matrix}
\acroplural{PCM}[PCMs]{parity check matrices}
\acro{PDF}{probability density function}

\acro{RS}{Reed--Solomon}
\acro{RV}[RV]{random variable}

\acro{SABM}{soft-aided bit marking}
\acro{SA-HDD}{soft-aided \ac{HDD}}
\acro{SABM-SR}{SABM with scaled reliabilities}
\acro{SDD}{soft decision decoding}
\acro{SISO}{soft input soft output}
\acro{SNR}[SNR]{signal-to-noise ratio}
\acro{SPC}[SPC]{single parity check}
\acro{TPD}{turbo product decoding}

\end{acronym}

\newcommand{\yone}{\bm{y}^{(1)}}
\newcommand{\ytwo}{\bm{y}^{(2)}}
\newcommand{\yi}{\bm{y}^{(i)}}
\newcommand{\wone}{\bm{w}^{(1)}}
\newcommand{\wtwo}{\bm{w}^{(2)}}
\newcommand{\wi}{\bm{w}^{(i)}}
\newcommand{\pone}{\bm{p}^{(1)}}
\newcommand{\ptwo}{\bm{p}^{(2)}}

\newcommand{\dmin}{d_{\textnormal{min}}}
\newcommand{\ddesign}{d_{\textup{des}}}
\newcommand{\CW}{\mathcal{C}}

\newcommand{\BDD}[1]{\textsf{BDD}\left(#1\right)}

\newcommand{\wc}{r}

\newcommand{\weight}[1]{\textnormal{wt}\left(#1\right)}
\newcommand{\hdist}[2]{d\left(#1,#2\right)}
\newcommand{\dnE}[1]{d_{\sim{\textnormal{E}(#1)}}}

\newcommand{\que}{\mathord{?}}
\newcommand{\ZQO}{\{0, \que, 1\}}

\newcommand{\pca}{P_{\textnormal{ca}}}

\newcommand{\pwa}{P_{\textnormal{wa}}}
\newcommand{\nca}{n_{\textnormal{ca}}}
\newcommand{\nwa}{n_{\textnormal{wa}}}
\newcommand{\ta}{\mathsf{T}_{\textnormal{a}}}
\newcommand{\te}{\mathsf{T}}

\newcommand{\PBDD}{P_{\textnormal{\textsf{BDD}}}}
\newcommand{\PBDDmc}{P_{\textnormal{\textsf{BDD}}}^{\textnormal{mc}}}
\newcommand{\PBDDfail}{P_{\textnormal{\textsf{BDD}}}^{\textnormal{fail}}}
\newcommand{\PBDDsucc}{P_{\textnormal{\textsf{BDD}}}^{\textnormal{succ}}}

\newcommand{\BDDa}{{\textnormal{\textsf{BDD}}}^{\textsf{a}}}
\newcommand{\PBDDamc}{P_{{\textnormal{\textsf{BDD}}}^{\textsf{a}}}^{\textnormal{mc}}}
\newcommand{\PBDDasucc}{P_{{\textnormal{\textsf{BDD}}}^{\textsf{a}}}^{\textnormal{succ}}}
\newcommand{\PBDDafail}{P_{{\textnormal{\textsf{BDD}}}^{\textsf{a}}}^{\textnormal{fail}}}

\newcommand{\PEaED}{P_{\textnormal{\textsf{EaED}}}}
\newcommand{\PEaEDmc}{P_{{\textnormal{\textsf{EaED}}}}^{\textnormal{mc}}}
\newcommand{\PEaEDfail}{P_{{\textnormal{\textsf{EaED}}}}^{\textnormal{fail}}}
\newcommand{\PEaEDsucc}{P_{{\textnormal{\textsf{EaED}}}}^{\textnormal{succ}}}

\newcommand{\PEaEDmcL}{P_{{\textnormal{\textsf{EaED}}}}^{\textnormal{mc}}}
\newcommand{\PEaEDfailL}{P_{{\textnormal{\textsf{EaED}}}}^{\textnormal{fail}}}
\newcommand{\PEaEDsuccL}{P_{{\textnormal{\textsf{EaED}}}}^{\textnormal{succ}}}

\newcommand{\PEaEDmcM}{P_{{\textnormal{\textsf{EaED}}}}^{\textnormal{mc}}}
\newcommand{\PEaEDfailM}{P_{{\textnormal{\textsf{EaED}}}}^{\textnormal{fail}}}

\newcommand{\EaEDa}{{\textnormal{\textsf{EaED}}}^{\textsf{a}}}
\newcommand{\PEaEDa}{P_{{\textnormal{\textsf{EaED}}}^{\textsf{a}}}}
\newcommand{\PEaEDamc}{P_{\textnormal{\textsf{EaED}}^{\textsf{a}}}^{\textnormal{mc}}}
\newcommand{\PEaEDafail}{P_{{\textnormal{\textsf{EaED}}^{\textsf{a}}}}^{\textnormal{fail}}}
\newcommand{\PEaEDasucc}{P_{{\textnormal{\textsf{EaED}}^{\textsf{a}}}}^{\textnormal{succ}}}

\newcommand{\PEaEDamcL}{P_{{\textnormal{\textsf{EaED}}^{\textsf{a}}}}^{\textnormal{mc}}}
\newcommand{\PEaEDafailL}{P_{{\textnormal{\textsf{EaED}}^{\textsf{a}}}}^{\textnormal{fail}}}
\newcommand{\PEaEDasuccL}{P_{{\textnormal{\textsf{EaED}}^{\textsf{a}}}}^{\textnormal{succ}}}

\newcommand{\PEaEDsamc}{P_{\textnormal{\textsf{EaEDs}}^{\textsf{a}}}^{\textnormal{mc}}}
\newcommand{\PEaEDsa}{P_{\textnormal{\textsf{EaEDs}}^{\textsf{a}}}}
\newcommand{\PEaEDsafail}{P_{\textnormal{\textsf{EaEDs}}^{\textsf{a}}}^{\textnormal{fail}}}

\usepackage{bm}
\usepackage{amsthm} 
\newtheorem{theorem}{Theorem}
\usepackage{tikz}
\usetikzlibrary{positioning}
\usetikzlibrary{decorations.pathreplacing}
\usepackage[
ruled,
vlined,
boxed, 
linesnumbered, 
commentsnumbered]{algorithm2e}
\usepackage{pifont}
\usepackage{booktabs}
\usepackage{mathtools}
\usepackage{diagbox} 
\usepackage{pgfplots}
\pgfplotsset{compat=newest}
\usepackage{amssymb}
\usepackage[colorlinks=true,bookmarks=false,citecolor=blue,urlcolor=blue,linkcolor=red]{hyperref}

\begin{document}
\title{On the Error Rate of Binary BCH Codes under Error-and-erasure Decoding}

 \author{%
   \IEEEauthorblockN{Sisi Miao, Jonathan~Mandelbaum, Holger Jäkel, and Laurent Schmalen}
   \IEEEauthorblockA{Karlsruhe Institute of Technology (KIT),
Communications Engineering Lab (CEL),
76187 Karlsruhe, Germany\\
Email: {\{\texttt{firstname.lastname@kit.edu\}}}}\vspace{-2em}}

\maketitle

\begin{abstract}
Determining the exact decoding error probability of linear block codes is an interesting problem. For binary BCH codes, McEliece derived methods to estimate the error probability of a simple bounded distance decoding (BDD) for BCH codes. However, BDD falls short in many applications. In this work, we consider error-and-erasure decoding and its improved variants. We derive closed-form expressions for their error probabilities and validate them through simulations. Then, we illustrate their use in assessing concatenated coding schemes.
\end{abstract}

\section{Introduction}
Binary \ac{BCH} codes with \ac{BDD} are widely used in communication systems~\cite{shulinbook}. To meet the growing performance demands, enhancing their coding gain through soft-decision decoding has become important. Decoders such as Chase~II~\cite{chase} and \ac{OSD}~\cite{osd_fossorier_lin_95} can approach \ac{ML} performance but at prohibitive complexity. In this paper, we consider \ac{EaE} decoding, perhaps the simplest soft-decision decoder for \ac{BCH} codes. Apart from the hard-decision values, \ac{EaE} decoding incorporates soft information by adding a single reliability class of erasures (“$\que$”)~\cite{ForneyGMDD}. There exist different variants of \ac{EaE} decoding.
In particular, we focus on the \ac{EaED} constructed from two \ac{BDD} steps~\cite[Sec.~3.8.1]{MoonBook}. This algorithm offers good decoding performance and is closely related to Chase~II decoding. Moreover, with a simple modification, it can reproduce the decoding result of Forney’s algebraic \ac{EaE} decoder~\cite{forney1965decoding,rapp2021error}.

One appealing property of \ac{BDD} is that its decoding error probability can be determined analytically~\cite{mceliece1986decoder}: given the number of input errors, the distribution of residual errors after decoding is known. In contrast, no such results exist yet for \ac{EaE} decoding. In this paper, we derive a closed-form expression for the considered \ac{EaED}. Specifically, we aim to determine the \ac{DTP} \({P(R\!=\!r | U\!=\!u, E\!=\!e)}\), where \(R\) is the number of residual errors after decoding, \(U\) the number of input errors, and \(E\) the number of input erasures. With the \ac{DTP}, the post-decoding \ac{BER} and \ac{FER} for various memoryless channels can be computed. The \ac{DTP} is validated through simulations and applied in a concatenated scheme.

\textbf{Notations}: We use boldface letters to denote vectors, e.g., $\bm{y}$, and $y_i$ denotes its $i$th component. Let  ${[n]:=\{0,1,\ldots,n-1\}}$. Furthermore, with ${\mathbb{F}_2=\{0,1\}}$, we introduce the Hamming weight ${\weight{\bm{y}}=|\{i\in[n]:y_i\neq 0\}|}$ for ${\bm{y}\in\mathbb{F}_2^n}$, the Hamming distance $\hdist{\bm{y}_1}{\bm{y}_2}=\weight{\bm{y}_1-\bm{y}_2}$, and the notation 
${\dnE{\bm{y}}(\bm{c}):=|\{i\in[n]:c_i\neq y_i,y_i\neq \que\}|}$ for ${\bm{c}\in \mathbb{F}_2^n}$ and $\bm{y}\in\{0,1,\que\}$ where ``$\que$'' is an erasure.
\Acp{RV} and their realizations are denoted by uppercase and lowercase letters, respectively, e.g., ${U=u}$.
Furthermore, when clear from the context, we use the shorthand notation $P(r|u)=P(R=r|U=u)$. Finally, $Q(x)\!=\!\tfrac{1}{\sqrt{2\pi}}\int_x^{\infty}\text{e}^{-\frac{u^2}{2}}\mathrm{d}u$ is the Gaussian tail distribution function.

\section{Preliminaries}

Let $\mathcal{C}[n, k, \ddesign]$ be a primitive, narrow-sense, binary \ac{BCH} code,
where $n$ is the block length, $k$ the dimension, and $\ddesign$ the designed distance, 
a lower bound on the minimum distance $\dmin$. 
The code rate is given by $R_c = k/n$. Let $t$ denote the number of correctable errors. The parameters are related by $n = 2^b-1$, $k \geq n - bt$, and $2t+1=\ddesign$.
Let $A_w$ denote the number of codewords of weight $w$ in $\mathcal{C}$. Binary primitive BCH codes have 
approximately binomial weight distributions  
given by \(A_w \approx 2^{-bt} \binom{n}{w}\) for \(d_{\mathrm{des}} \leq w \leq n - d_{\mathrm{des}}\), with \(A_0 = A_n = 1\), and \(A_w = 0\) elsewhere~\cite{shulinbook}.
The exact weight enumerators of some short BCH codes are listed in~\cite{OEISWeightDistributions}. Note that both~\cite{mceliece1986decoder} and the results derived in this work rely on the weight distribution. If only an approximate weight distribution is available, the resulting analysis is also approximate.

We consider three communication channels: the \ac{BSC} with cross-over probability ${\delta=P(Y\neq X)}$, the \ac{EaE} channel, and the \ac{BI-AWGN} channel. 
In the \ac{EaE} channel, the input is ${X\in \mathbb{F}_2}$ while the output is ${Y\in\{0,1,\que\}}$. The channel transition probabilities are ${P(Y=1-x|X=x)=\delta_c}$, ${P(Y=\que|X=x)=\epsilon_c}$, and ${P(Y=x|X=x)=1-\delta_c-\epsilon_c}.$ For the \ac{BI-AWGN} channel, we map bit ${x_i\mapsto(-1)^{x_i}=:\tilde{x}_i}$, and the received symbol is ${\tilde{y}_i = \tilde{x}_i + n_i\in \mathbb{R}}$, where ${n_i\sim\mathcal{N}(0,\sigma_n^2)}$ with $\sigma_n^2 = (2R_cE_b/N_0)^{-1}$.
In the following, by introducing a threshold ${\te\in \mathbb{R}}$, we reduce the \ac{BI-AWGN} to an \ac{EaE} channel: if ${\tilde{y}_i \in [-\te,\te]}$, an erasure is declared, i.e., $y_i=\que$, otherwise ${y_i=\mathrm{sign}(\tilde{y}_i)}$. 
Thus, the parameters of the resulting \ac{EaE} are
$\delta_c =  Q\left(\tfrac{\te+1}{\sigma}\right)$, and $\epsilon_c = 1 - Q\left(\tfrac{\te-1}{\sigma}\right) - \delta_c.$ For $\te=0$, we recover a \ac{BSC} with $\delta = Q\left(\tfrac{1}{\sigma}\right)$.

We use $\bm{x}\in \mathcal{C}$ to denote the transmitted codeword and $\bm{y}$ to denote the received word, whose alphabet is determined by the channel.  At the output of a \ac{BSC}, we have 
\[
\BDD{\bm{y}} =
\begin{cases}
    \bm{c}, & \text{if } \bm{y} \in \mathcal{S}_t(\bm{c}),\\
    \bm{y}\,(\mathsf{fail}), & \text{otherwise},
\end{cases}
\]
where $\mathcal{S}_t(\bm{c})=\{\bm{y}\in \mathbb{F}_2^n:d(\bm{y},\bm{c})\leq t\}$ denote the Hamming sphere of radius $t$ for $\bm{c}\in \mathcal{C}$. The decoding outcome falls into one of three categories:  
(i) decoding success, if $\bm{c} = \bm{x}$;  
(ii) decoding failure if no codeword is found; or  
(iii) miscorrection (also called a decoding error as in~\cite{mceliece1986decoder}), if $\bm{c} \neq \bm{x}$.

For each decoder considered in this paper, we distinguish these three cases when computing the \ac{DTP}. Taking \ac{BDD} as an example, in addition to the \acp{RV} $U$ and $R$, we introduce the \ac{RV} ${D\in\mathcal{D}:=\{\mathsf{succ},\mathsf{fail},\mathsf{mc}\}}$ representing the decoding outcome for success, failure, and miscorrection, respectively. Then, we define the conditional joint probability ${\PBDD(R=r, D=d | U=u)}$, i.e., the probability of observing outcome $d$ with $r$ residual errors given $u$ input errors. For simplicity, we use the shorthand notation $\PBDD^{d}(r|u):={\PBDD(r,d|u)}$, yielding $\PBDDsucc(r|u)$, $\PBDDfail(r|u)$, and $\PBDDmc(r|u)$.
It follows that $$\PBDD(r|u) = \PBDDmc(r|u) + \PBDDfail(r|u) + \PBDDsucc(r|u).$$
Note that for a decoding failure where ${\BDD{\bm{y}} = \bm{y}}$, ${\PBDDfail(R \neq u | u) \! = \!0}$, while for a decoding success, ${\PBDDsucc(R \neq 0 | u) \!=\! 0}$. We further denote by $\PBDDsucc(u)$, $\PBDDfail(u)$, and $\PBDDmc(u)$ the conditional probabilities of success, failure, and miscorrection, respectively, given $u$ errors.

Since we study linear codes with a symmetric decoding algorithm, we assume the transmission of ${\bm{x}=\bm{0}}$ throughout the paper. Then, the decoder input errors are ${u=\weight{\bm{y}}}$.

\section{Analysis of Bounded Distance Decoding}

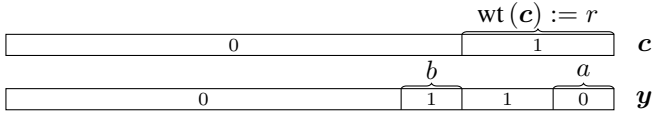
\begin{figure}
    \centering

\begin{tikzpicture}[scale=0.4]
  \filldraw[draw=black, fill=white] (0,1.5) rectangle (20,2.2);   %
  \filldraw[draw=black, fill=white] (0,3.3) rectangle (20,4.0);   %

  \node at (21,1.85) {$\bm{y}$};
  \node at (21,3.65) {$\bm{c}$};

  \draw (15,3.3) -- (15,4.0);
  \draw (15,1.5) -- (15,2.2);
  \draw (13,1.5) -- (13,2.2);
  \draw (18,1.5) -- (18,2.2);
  
  \node[font=\scriptsize] at (7.5,3.65) {$0$};
  \node[font=\scriptsize] at (17.5,3.65) {$1$};
  
  \node[font=\scriptsize] at (6.5,1.85) {$0$};
  \node[font=\scriptsize] at (14,1.85) {$1$};
  \node[font=\scriptsize] at (16.5,1.85) {$1$};
  \node[font=\scriptsize] at (19,1.85) {$0$};

  \draw[decorate,decoration={brace,amplitude=2pt}]
    (15,4.0) -- (20,4.0)
    node[midway,yshift=7pt]{$\weight{\bm{c}}:=r$};
  \draw[decorate,decoration={brace,amplitude=2pt}]
    (13,2.2) -- (15,2.2)
    node[midway,yshift=7pt]{$b$};
  \draw[decorate,decoration={brace,amplitude=2pt}]
    (18,2.2) -- (20,2.2)
    node[midway,yshift=7pt]{$a$};
\end{tikzpicture}
\vspace{-1ex}
    \caption{Graphical illustration of the miscorrection scenario when using a BDD.}
    \vspace{-2ex}
    \label{fig:BDD}
\end{figure}

\label{sec:BDD}

In this section, we consider \ac{BDD} after \ac{BSC} where the received sequence is $\bm{y}\in \mathbb{F}_2^n$, i.e., without erasures.
To this end, we review the method in~\cite{mceliece1986decoder} and refine it slightly such that it also gives the detailed \ac{DTP} for \ac{BDD}. 

\textbf{Case 1:} If $U=u\leq t$, then $\PBDDsucc(0|u)=1$ and $\PBDD^{d}(r|u)=0$ for $d\neq\mathsf{succ},r\neq0$.

\textbf{Case 2:} If ${U=u> t}$, either a miscorrection or a decoding failure may occur. In the miscorrection scenario, we have $\PBDDmc(r | u) = N(u)/\tbinom{n}{u}$, where $N(u)$ is the number of received patterns $\bm{y}$ with $\weight{\bm{y}}=u$ that satisfy $\bm{y} \in \mathcal{S}_t(\bm{c})$ for some $\bm{c} \in \mathcal{C}$ with $\weight{\bm{c}}=r$.  Figure~\ref{fig:BDD} illustrates this scenario with the support of $\bm{c}$ permuted to the right.

We define
${a = |\{i\in[n]:c_i=1,y_i=0\}|}$ as the number of correct bits in $\bm{y}$ that are (wrongly) flipped by \ac{BDD}, and ${b = |\{i\in[n]:c_i=0,y_i=1\}|}$ as the number of erroneous positions in $\bm{y}$ that are (correctly) flipped by \ac{BDD}.
Hence, it holds that ${u+a-b=r}$ and $a+b = \hdist{\bm{y}}{\bm{c}}\leq t$. 
For fixed $a$ and $b$, there are $\binom{r}{a}\binom{n-r}{b}=\binom{r}{a}\binom{n-r}{u+a-r}$ distinct patterns for $\bm{y}$. 
Additionally, the number of codewords $\bm{c}\in\mathcal{C}$ with $\weight{\bm{c}}=r,$ is given by the weight enumerator $A_r$.
Thus, 
    \begin{equation*}
        \PBDDmc(r|u) = \frac{N(u)}{\tbinom{n}{u}} = \frac{\sum_{(a,b)\in \mathcal{S}_{a,b}}  A_r \tbinom{r}{a}\tbinom{n-r}{b}}{\tbinom{n}{u}}.
    \end{equation*}
    where ${\mathcal{S}_{a,b} = \{(a,b) \in [t]\times [t]: a+b \leq t, u+a-b=r\}}$.
Then, as ${\PBDDmc(r|u)\! +\! \PBDDfail(r|u) \!= \!1}$, we have \[\PBDDfail(r|u) = 1-\sum_{\tilde{r}=u-t}^{u+t} \PBDDmc(\tilde{r}|u).\] The results are summarized in Theorem~\ref{theorem:BDD}.

\begin{theorem}\label{theorem:BDD}
The DTP of \ac{BDD} is given as follows. For $u \le t$, $\PBDD(0|u) = 1$, $\PBDD(r|u) = 0$, $(d,r)\neq (\mathsf{succ},0)$.
For $u > t$,
    \begin{equation*}
         \PBDD(r|u) \!\!=\! \!\begin{cases}
            \PBDDmc(r|u), & \!\!\!u-t\!\leq \! r \! \leq \! u+t, \! r\!\neq \!u,\\
            \PBDDmc(r|u)\!+\!\PBDDfail(r|u), &\!\!\! r=u,\\
            0,   & \!\!\!\text{otherwise.}
        \end{cases}
    \end{equation*}
\end{theorem}

\section{Error-and-erasure Decoding}

\subsection{Error-and-erasure (EaE) Decoding Algorithm}\label{sec:EaEDalgorithm}

In the following, we revisit the \ac{EaED} algorithm presented in \cite[Algorithm~2]{miao2022JLT}. For a received word $\bm{y} \in \ZQO^n$ with $E$ erasures and $U$ errors, let ${\bm{w} := \textnormal{EaED}(\bm{y})}$. 

In EaED, a pair of random \emph{filling patterns} $(\pone,\ptwo)$ with 
${\pone,\ptwo \in \mathbb{F}_2^{E}}$ and ${\pone+\ptwo=\bm{1}_{E}}$ is constructed, 
by drawing $\pone$ as a length-$E$ Bernoulli($\frac{1}{2}$) vector. The erasures in $\bm{y}$ are then replaced by $\pone$ and $\ptwo$ to form the \emph{test patterns} $\yone, \ytwo \in \mathbb{F}_2^n$. Two \acp{BDD}  yield $\wi := \BDD{\yi}$, $i\in\{1,2\}$. The final output $\bm{w}$ is selected from $\wone$ and $\wtwo$ according to the following rules:

\textbf{Case 1:} If both \ac{BDD} fail, set $\bm{w}=\bm{y}$ and declare a failure. 

\textbf{Case 2:} If $\wi\in \CW$ for exactly one $\wi$, set $\bm{w}=\wi$. 

\textbf{Case 3:} If both BDD succeed, let ${d_i=\dnE{\bm{y}}\left(\wi\right)}$ for ${i\in\{1,2\}}$.
Then,  choose ${\bm{w}=\wone}$ if ${d_1<d_2}$, and ${\bm{w}=\wtwo}$ if $d_2<d_1$. If $d_1 = d_2$, the output $\bm{w}$ is chosen as $\wone$ or $\wtwo$ with equal probability.

\textbf{Remark 1}: To avoid miscorrections, one commonly declares a decoding failure if the number of erasures $E$ exceeds a certain value before decoding.
This can be easily incorporated in the upcoming \ac{DTP} derivation by truncation.

\textbf{Remark 2}: The EaED reverts to a conventional BDD when there are no erasures, i.e., $e=0$.

\textbf{Remark 3}: For simplicity, when EaED declares a decoding failure, each residual erasure in $\bm{y}$ is resolved by an independent Bernoulli($\frac{1}{2}$) trial, such that on average every residual erasure contributes to $\frac{1}{2}$ bit error.

\subsection{Decoding Transition Probability}\label{sec:eaed_P}

Theorem~\ref{theorem:EaED} outlines our derived  $\PEaED(r |u,e)$ of \ac{EaED}.
\begin{theorem}\label{theorem:EaED}
The \ac{DTP} of \ac{EaED} is given as follows. For ${2u+e<\dmin}$, ${\PEaED(0|u,e)=1}$. For ${2u+e \geq \dmin}$, $\PEaED(r|u,e) = \sum_{d\in\mathcal{D}}\PEaED^d(r|u,e)$ and
\begin{equation}
    \PEaED^d(r|u,e)=\textstyle \sum_{e_1=0}^{e}\PEaED^d(r|u,e,e_1)\cdot P(e_1|u,e),\label{eq:p_dr_ue}
\end{equation}
where ${e_1 = \!\weight{\pone}}$. The value $\PEaED^d(r|u,e,e_1)$ is given by \eqref{eq:probinL} when ${e_1\in[0,t-u]}$. For ${e_1\in[u+e-t,e]}$, we have the symmetry relation ${\PEaED^d(r|u,e,e_1)=\PEaED^d(r|u,e,e\!-\!e_1)}$. For ${e_1\in(t-u,u+e-t)}$, it is approximated by \eqref{eq:caseMfail} and \eqref{eq:caseM}.
\end{theorem}
\begin{proof}
    See Appendix~\ref{appendx:EaED}.
\end{proof}

Note that, due to the Bernoulli trials to resolve residual erasures upon decoding failure (Remark~3), $R$ can now take values in $\{0, \tfrac{1}{2}, 1, \tfrac{3}{2}, \dots, n\}$.

\begin{table}[t]    
\caption{Probability of success of EaED with $u$ errors and $e$ erasures for the $[255,239,5]$ primitive BCH code.}
\label{tab:PsuccEaED}
\centering
\setlength{\tabcolsep}{3pt}
\begin{tabular}{c|ccccccccc}
\toprule
\diagbox[height=1.5em]{$u$}{$e$} &0& 1 & 2 & 3 & 4 & 5 & 6 & 7 & 8 \\
\midrule
0 & 1.000 & 1.000 & 1.000 & 1.000 & 1.000 & 0.999992 & 0.688 & 0.453 & 0.289 \\
1 & 1.000 & 1.000 & 1.000 & 0.998 & 0.622 & 0.371 & 0.216 & 0.123 & 0.069 \\
2 & 1.000 & 0.753 & 0.376 & 0.186 & 0.093 & 0.046 & 0.023 & 0.012 & 0.006 \\
\bottomrule
\end{tabular}
\end{table}

\begin{table}[t]
\caption{Probability of miscorrection of EaED with $u$ errors and $e$ erasures for the $[255,239,5]$ primitive BCH code.}
\label{tab:PmissEaED}
\centering
\setlength{\tabcolsep}{3pt}
\begin{tabular}{c|ccccccccc}
\toprule
\diagbox[height=1.5em]{$u$}{$e$} &0& 1 & 2 & 3 & 4 & 5 & 6 & 7 & 8 \\
\midrule
0 & 0 & 0 & 0 & 0 & 0 & 7.8e-06 & 0.233 & 0.407 & 0.530 \\
1 & 0 & 0 & 0 & 0.002 & 0.282 & 0.469 & 0.585 & 0.655 & 0.695 \\
2 & 0 & 0.247 & 0.497 & 0.622 & 0.684 & 0.716 & 0.732 & 0.740 & 0.744 \\
3 & 0.494 & 0.744 & 0.745 & 0.746 & 0.747 & 0.747 & 0.748 & 0.748 & 0.748 \\
4 & 0.494 & 0.746 & 0.747 & 0.748 & 0.748 & 0.748 & 0.748 & 0.748 & 0.748 \\
5 & 0.498 & 0.748 & 0.748 & 0.748 & 0.748 & 0.748 & 0.748 & 0.748 & 0.748 \\
\bottomrule
\end{tabular}
\end{table}

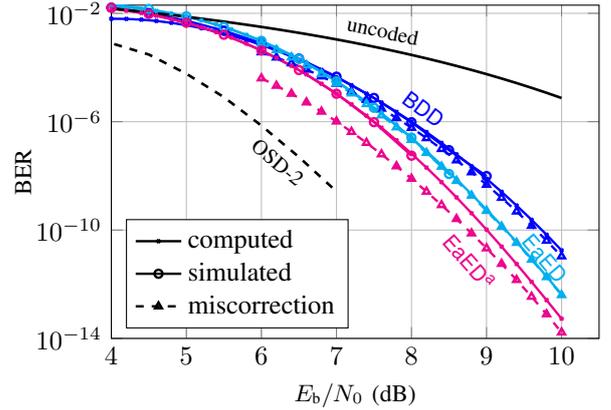
\begin{figure}[t]
    \centering
    \begin{tikzpicture}
\begin{axis}[
    width=8cm,
    height=6cm,
    grid=both,
    xlabel={$E_{\text{b}}/N_0$ (dB)},
    label style={font=\small},
    ylabel style={yshift=-0.18cm},
    ylabel={BER},
    ymode=log,
    legend pos=south west,
    legend cell align={left},
    ymax = 2e-2,
    ymin = 1e-14,
    xmin=4,
    xmax=10.5
]

\addplot [
  color=black,
  line width=0.9pt,
  mark=x,
  mark size=0.9pt,
  mark options={solid}
]
table {
10.000000 1e-18
};
\addlegendentry{computed}

\addplot [
  color=black,
  line width=0.9pt,
  mark=o,
  mark size=1.5pt,
  mark options={solid}
]
table {
10.000000 1e-18
};
\addlegendentry{simulated}

\addplot [
  color=black,
  line width=0.9pt,
  mark=triangle,
  dashed,
  mark size=1.5pt,
  mark options={solid}
]
table {
10.000000 1e-18
};
\addlegendentry{miscorrection}

\addplot [
  color=blue,
  line width=0.9pt,
  mark=x,
  mark size=0.9pt,
  mark options={solid}
]
table {
4.000000 6.301635e-03
4.200000 6.218244e-03
4.400000 5.842282e-03
4.600000 5.231965e-03
4.800000 4.470321e-03
5.000000 3.647391e-03
5.200000 2.844009e-03
5.400000 2.120724e-03
5.600000 1.513260e-03
5.800000 1.033881e-03
6.000000 6.766906e-04
6.200000 4.245142e-04
6.400000 2.553755e-04
6.600000 1.473771e-04
6.800000 8.161898e-05
7.000000 4.338795e-05
7.200000 2.214212e-05
7.400000 1.084768e-05
7.600000 5.100936e-06
7.800000 2.301499e-06
8.000000 9.958465e-07
8.200000 4.129372e-07
8.400000 1.639417e-07
8.600000 6.224811e-08
8.800000 2.257496e-08
9.000000 7.808003e-09
9.200000 2.571176e-09
9.400000 8.046174e-10
9.600000 2.387939e-10
9.800000 6.706082e-11
10.000000 1.777828e-11
};

\addplot [
  color=blue,
  line width=0.9pt,
  mark=o,
  mark size=1.5pt,
  mark options={solid}
]
table {
4 0.0162452
4.5 0.010157
5 0.00535121
5.5 0.00229383
6 0.000786534
6.5 0.000214438
7 4.59055e-05
7.5 7.52667e-06
8 9.54118e-07
8.5 8.90196e-08
9 9.80392e-09
}node [pos=0.7,anchor=south,font=\footnotesize,sloped] {\textsf{BDD}};

\addplot [
  color=blue,
  line width=0.9pt,
  mark=triangle,
  dashed,
  mark size=1.5pt,
  mark options={solid}
]
table {
6.000000 3.657121e-04
6.200000 2.344468e-04
6.400000 1.435976e-04
6.600000 8.411066e-05
6.800000 4.715250e-05
7.000000 2.531593e-05
7.200000 1.302360e-05
7.400000 6.421678e-06
7.600000 3.035225e-06
7.800000 1.375028e-06
8.000000 5.968550e-07
8.200000 2.480989e-07
8.400000 9.868353e-08
8.600000 3.752297e-08
8.800000 1.362251e-08
9.000000 4.715289e-09
9.200000 1.553622e-09
9.400000 4.863826e-10
9.600000 1.443892e-10
9.800000 4.055694e-11
10.000000 1.075334e-11
};

\addplot [
  color=cyan,
  line width=0.9pt,
  mark=x,
  mark size=0.9pt,
  mark options={solid}
]
table {
4.000000 1.825301e-02
4.200000 1.661692e-02
4.400000 1.453128e-02
4.600000 1.218998e-02
4.800000 9.792818e-03
5.000000 7.519804e-03
5.200000 5.508999e-03
5.400000 3.843279e-03
5.600000 2.548813e-03
5.800000 1.604325e-03
6.000000 9.570996e-04
6.200000 5.405279e-04
6.400000 2.887108e-04
6.600000 1.457435e-04
6.800000 6.950394e-05
7.000000 3.130829e-05
7.200000 1.332302e-05
7.400000 5.358466e-06
7.600000 2.038542e-06
7.800000 7.344149e-07
8.000000 2.509360e-07
8.200000 8.146787e-08
8.400000 2.518386e-08
8.600000 7.428802e-09
8.800000 2.095308e-09
9.000000 5.659236e-10
9.200000 1.464584e-10
9.400000 3.629724e-11
9.600000 8.598516e-12
9.800000 1.940944e-12
10.000000 4.157751e-13
}node [pos=0.9,anchor=south,font=\footnotesize,sloped, text = cyan, yshift=-0.1cm] {\textsf{EaED}};

\addplot [
  color=cyan,
  line width=0.9pt,
  mark=o,
  mark size=1.5pt,
  mark options={solid}
]
table {
4 0.0222144
4.5 0.0144888
5 0.00772447
5.5 0.00317196
6 0.000960768
6.5 0.000207601
7 3.10398e-05
7.5 3.26804e-06
8 2.57059e-07
8.5 1.17647e-08
};

\addplot [
  color=cyan,
  line width=0.9pt,
  mark=triangle,
  dashed,
  mark size=1.5pt,
  mark options={solid}
]
table {
6.000000 7.811834e-04
6.200000 4.422113e-04
6.400000 2.368033e-04
6.600000 1.198829e-04
6.800000 5.735769e-05
7.000000 2.593429e-05
7.200000 1.108473e-05
7.400000 4.481354e-06
7.600000 1.715312e-06
7.800000 6.224388e-07
8.000000 2.144800e-07
8.200000 7.031550e-08
8.400000 2.197823e-08
8.600000 6.562835e-09
8.800000 1.875334e-09
9.000000 5.133256e-10
9.200000 1.346006e-10
9.400000 3.377078e-11
9.600000 8.088208e-12
9.800000 1.842878e-12
10.000000 3.977820e-13
};

\addplot [
  color=magenta,
  line width=0.9pt,
  mark=x,
  mark size=0.9pt,
  mark options={solid}
]
table {
4.000000 1.469172e-02
4.200000 1.255131e-02
4.400000 1.031028e-02
4.600000 8.134703e-03
4.800000 6.154316e-03
5.000000 4.455806e-03
5.200000 3.080927e-03
5.400000 2.030356e-03
5.600000 1.272908e-03
5.800000 7.579661e-04
6.000000 4.280852e-04
6.200000 2.290575e-04
6.400000 1.160114e-04
6.600000 5.557744e-05
6.800000 2.517220e-05
7.000000 1.077515e-05
7.200000 4.358368e-06
7.400000 1.665709e-06
7.600000 6.015635e-07
7.800000 2.053351e-07
8.000000 6.626870e-08
8.200000 2.023372e-08
8.400000 5.849932e-09
8.600000 1.603596e-09
8.800000 4.175451e-10
9.000000 1.035308e-10
9.200000 2.452625e-11
9.400000 5.574178e-12
9.600000 1.221014e-12
9.800000 2.590349e-13
10.000000 5.345759e-14
}node [pos=0.75,anchor=north,font=\footnotesize,sloped, text = magenta, yshift=-0.2cm] {$\EaEDa$};

\addplot [
  color=magenta,
  line width=0.9pt,
  mark=o,
  mark size=1.5pt,
  mark options={solid}
]
table {
4 0.016003
4.5 0.00947019
5 0.00449967
5.5 0.00162124
6 0.000429778
6.5 8.05565e-05
7 1.07651e-05
7.5 9.61373e-07
8 5.56863e-08
};

\addplot [
  color=magenta,
  line width=0.9pt,
  mark=triangle,
  dashed,
  mark size=1.5pt,
  mark options={solid}
]
table {
6.000000 3.998556e-05
6.200000 2.129299e-05
6.400000 1.075904e-05
6.600000 5.160927e-06
6.800000 2.352291e-06
7.000000 1.020044e-06
7.200000 4.215075e-07
7.400000 1.662796e-07
7.600000 6.273924e-08
7.800000 2.268134e-08
8.000000 7.867435e-09
8.200000 2.620476e-09
8.400000 8.381444e-10
8.600000 2.572046e-10
8.800000 7.559906e-11
9.000000 2.123071e-11
9.200000 5.679487e-12
9.400000 1.442351e-12
9.600000 3.464792e-13
9.800000 7.843569e-14
10.000000 1.667074e-14
};
\addplot [
  color=black,
  line width=1pt,
  mark=none
]
table {
4.000000 1.500643e-02
4.200000 1.319362e-02
4.400000 1.153764e-02
4.600000 1.003299e-02
4.800000 8.673466e-03
5.000000 7.452255e-03
5.200000 6.361998e-03
5.400000 5.394889e-03
5.600000 4.542771e-03
5.800000 3.797245e-03
6.000000 3.149771e-03
6.200000 2.591782e-03
6.400000 2.114784e-03
6.600000 1.710461e-03
6.800000 1.370766e-03
7.000000 1.088006e-03
7.200000 8.549174e-04
7.400000 6.647195e-04
7.600000 5.111659e-04
7.800000 3.885724e-04
8.000000 2.918339e-04
8.200000 2.164257e-04
8.400000 1.583930e-04
8.600000 1.143271e-04
8.800000 8.133320e-05
9.000000 5.698999e-05
9.200000 3.930372e-05
9.400000 2.665943e-05
9.600000 1.777104e-05
9.800000 1.163231e-05
10.000000 7.470346e-06
}node [pos=0.5,anchor=south,font=\footnotesize,sloped] {uncoded};

\addplot [
  color=black,
  dashed,
  line width=1pt,
  mark=none
]
table {
 3.00  3.087849e-03
 3.50  2.154708e-03
 4.00  7.898426e-04
 4.50  3.009646e-04
 5.00  5.895766e-05
 5.50  8.055728e-06
 6.00  7.157440e-07
 6.50  4.962827e-08
 7.00  2.653639e-09
}node [pos=0.75,anchor=north,font=\footnotesize,sloped] {OSD-2};

\end{axis}
\end{tikzpicture}
    \caption{Simulated and computed \ac{BER} results of the  $[255,239,5]$ primitive BCH code. The $\mathsf{EaED}$ uses the parameter $\te = 0.16$ and the $\EaEDa$ uses the parameters $\te = 0.13$ and $\ta=0.75$. The \ac{OSD} uses order~2, as higher orders do not yield additional performance gains.}
    \label{fig:oneshot}
\end{figure}

Using the $[255,239,5]$ BCH code as an example, Tables~\ref{tab:PsuccEaED} and~\ref{tab:PmissEaED} list the conditional probabilities $\PEaEDsucc(u,e)$ and $\PEaEDmc(u,e)$, respectively. These quantities represent the probabilities of successful decoding and miscorrection given $u$ errors and $e$ erasures, computed using~\eqref{eq:p_dr_ue} and marginalized over $R$. Note that if $e=0$, Theorem~\ref{theorem:EaED} reverts to Theorem~\ref{theorem:BDD}.
Accordingly, the first column of both tables corresponds to BDD. 
The miscorrection rate of BDD is high and is close to the upper bound $\frac{1}{t!}$ derived in~\cite{mceliece1986decoder}.
Compared with BDD, EaED can decode beyond the minimum distance. 
When $2u+e$ slightly exceeds $\dmin-1$, the decoding success rate of \ac{EaED} is still high. However, this comes at the expense of an increased probability of miscorrection for larger values of $2u+e$. The post-decoding \acp{BER} of EaED and BDD are compared in Fig.~\ref{fig:oneshot} with the threshold $\te$ optimized to minimize the \ac{BER} of EaED at $9$ dB. We see that the simulation and computed values closely agree. Additionally, as expected, EaED improves upon BDD in terms of performance.

\section{Miscorrection Reduction with Anchor Bits}

Miscorrections cause undetectable errors that can be more harmful than decoding failures, e.g., when retransmission is allowed upon decoding failure, or during iterative decoding of product codes with BCH component codes, where miscorrections lead to error propagation.

To mitigate miscorrections, we consider \emph{anchor bits}~\cite{hager2018approaching,miao2022JLT}, which form an additional reliability class. In contrast to erasures, anchor bits are assumed to be highly reliable, and any decoding decision that conflicts with them is treated as a miscorrection. Let $\mathcal{H}_{\mathsf{a}} \subseteq [n]$ denote the indices of anchor bits. For a decoding result $\bm{w} = \mathsf{Dec}(\bm{y})$ with $\bm{w} \in \mathcal{C}$, we perform a \ac{MD} step: if there exists $i \in \mathcal{H}_{\mathsf{a}}$ such that $w_i \neq y_i$, the estimate $\bm{w}$ is rejected and the decoder outputs a failure. It remains to define the set of anchor bits $\mathcal{H}_{\mathsf{a}}$. In this work, anchor bits are obtained by applying a threshold $\ta$ to the \ac{BI-AWGN} channel output, i.e., by marking bits with output value $|\tilde{y}_i|>\ta$. Therefore, $\mathcal{H}_{\mathsf{a}} = \{i\in[n]:|\tilde{y}_i|>\ta\}$. Next, we show that the framework summarized in Theorems~\ref{theorem:BDD} and~\ref{theorem:EaED} can be extended to this scenario with minor modifications.

\subsection{BDD with Anchor Bits}
In this section, we present the \ac{DTP} of BDD aided by anchor bits, termed a $\BDDa$ decoder. We define %
\begin{align*}
    &\pca = P(\tilde{Y}>\ta | \tilde{Y}>0,X=1) = Q\!\left(\tfrac{\ta-1}{\sigma}\right)/(1-\delta),\\
    &\pwa = P(\tilde{Y}<-\ta | \tilde{Y}<0,X=1) = Q\!\left(\tfrac{-\ta-1}{\sigma}\right)/\delta
\end{align*}
to be the probability of a correct and erroneous anchor bit, respectively.
Following the analysis of Sec.~\ref{sec:BDD} and the illustration in Fig.~\ref{fig:BDD}, a miscorrection $\BDD{\bm{y}}=\bm{c}$ is rejected by the anchor bits with probability $(1-\pca)^a  (1-\pwa)^b$. Thus,
\[\PBDDamc(r|u) =\sum_{(a,b)\in \mathcal{S}_{a,b}} A_r \tbinom{r}{a}\tbinom{n-r}{b} (1-\pca)^a  (1-\pwa)^b / \tbinom{n}{u} ,\]
where $(1-\pca)^a\ll (1-\pwa)^b$ since typically $\pca \gg \pwa$. 
$\PBDDafail$ and $\PBDDasucc$ can then be calculated easily.

\subsection{EaED with Anchor Bits}

Similar to BDD$^{\mathsf{a}}$, we introduce EaED$^{\mathsf{a}}$, which is an EaED with the \ac{MD} step using anchor bits, whose \ac{DTP} is denoted by $\PEaEDa(r|u,e)$. To be precise, in EaED$^{\mathsf{a}}$, for $\bm{w}^{(i)}=\BDD{\bm{y}^{(i)}}, i\in \{1,2\}$, we check if there exists $j\in \mathcal{H}_{\mathsf{a}}$ such that $w^{(i)}_ j\neq y_j$. If this is true, we declare a decoding failure for $\BDD{\bm{y}^{(i)}}$. Then, we proceed to select the final output as described in Sec.~\ref{sec:EaEDalgorithm}.

Note that we always set $\ta>\te$ to ensure that erasures cannot be anchor bits.
In this case, $\pca$ and $\pwa$ are given as
\begin{align*}
    &\pca = P(\tilde{Y}>\ta | \tilde{Y}>\te,X=1) = Q\left(\tfrac{\ta-1}{\sigma}\right)/(1-\delta_c-\epsilon_c),\\
    &\pwa=P(\tilde{Y}<\ta | \tilde{Y}<\te,X=1)=Q\left(\tfrac{-\ta-1}{\sigma}\right)/\delta_c.
\end{align*}
The \ac{DTP} of  EaED$^{\mathsf{a}}$ is calculated in Appendix~\ref{appendx:EaEDa}.

\subsection{Improving one-shot decoding of BCH code}

As shown in Fig.~\ref{fig:oneshot}, EaED$^{\mathsf{a}}$ outperforms both EaED and BDD. 
Next, we examine the miscorrection rate. 
For EaED, most bit errors are due to miscorrections, whereas BDD exhibits a smaller fraction.
This observation aligns with Tab.~\ref{tab:PsuccEaED} and Tab.~\ref{tab:PmissEaED}, reflecting that EaED has a stronger error correction capability but is more prone to miscorrection.
By contrast, EaED$^{\mathsf{a}}$ substantially reduces miscorrections through the use of anchor bits, which can be seen from the dashed lines in Fig.~\ref{fig:oneshot}.
Thus, EaED$^{\mathsf{a}}$ improves decoding performance while maintaining a low miscorrection rate, keeping the undetectable error rate negligible.

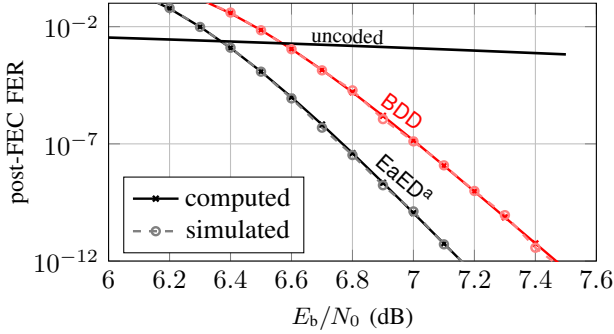
\begin{figure}
    \centering
    \begin{tikzpicture}

\begin{axis}[
    ylabel={post-FEC FER},
    ymode=log,             
    xlabel={$E_{\text{b}}/N_0$ (dB)},
    label style={font=\small},
    ylabel style={yshift=-0.18cm},
    grid=both,
    width=8cm,
    height=5cm,
    legend pos=south west,
    ymax=1e-1,
    ymin = 1e-12,
    xmin = 6,
    xmax=7.6
]
\addlegendentry{computed}
\addlegendentry{simulated}

\addplot [
  color=black,
  line width=0.9pt,
  mark=x,
  mark size=1.5pt,
  mark options={solid},
]
table {
  6.000000  6.269077e-01
  6.100000  2.457441e-01
  6.200000  5.920345e-02
  6.300000  9.893023e-03
  6.400000  1.239258e-03
  6.500000  1.219475e-04
  6.600000  9.798819e-06
  6.700000  6.677286e-07
  6.800000  4.004161e-08
  6.900000  2.187670e-09
  7.000000  1.123448e-10
  7.100000  5.570655e-12
  7.200000  2.726708e-13
  7.300000  1.332268e-14
  7.400000  8.881784e-16
};

\addplot [
  color=black!50,
  line width=0.9pt,
  mark=o,
  mark size=1.5pt,
  mark options={solid},
  dashed
]
table {
6.000000 6.105257e-01
6.100000 2.441153e-01
6.200000 6.254050e-02
6.300000 9.716704e-03
6.400000 1.263699e-03
6.500000 1.201537e-04
6.600000 8.687288e-06
6.700000 4.912309e-07
6.800000 3.333023e-08
6.900000 1.715479e-09
7.000000 1.312044e-10
7.100000 5.177192e-12
7.200000 3.028688e-13
7.300000 8.881784e-15
7.400000 8.881784e-16
}node [pos=0.6,anchor=south,font=\footnotesize,sloped, text = black] {$\EaEDa$};

\addplot [
  color=red,
  line width=0.9pt,
  mark=x,
  mark size=1.5pt,
  mark options={solid},
]
table {
6.000000 9.683738e-01
6.100000 7.789450e-01
6.200000 4.220445e-01
6.300000 1.489278e-01
6.400000 3.724071e-02
6.500000 7.152432e-03
6.600000 1.108871e-03
6.700000 1.435039e-04
6.800000 1.596217e-05
6.900000 1.570210e-06
7.000000 1.404612e-07
7.100000 1.173106e-08
7.200000 9.370638e-10
7.300000 7.312995e-11
7.400000 5.678125e-12
7.500000 4.449774e-13
7.600000 3.552714e-14
7.700000 2.664535e-15
};

\addplot [
  color=red!50,
  line width=0.9pt,
  mark=o,
  mark size=1.5pt,
  mark options={solid},
  dashed
]
table {
6.000000 9.693130e-01
6.100000 7.982895e-01
6.200000 4.355290e-01
6.300000 1.483068e-01
6.400000 4.004526e-02
6.500000 7.013516e-03
6.600000 1.095085e-03
6.700000 1.370160e-04
6.800000 1.918109e-05
6.900000 1.137620e-06
7.000000 1.275335e-07
7.100000 1.215024e-08
7.200000 9.769314e-10
7.300000 9.109069e-11
7.400000 3.713474e-12
7.500000 3.632650e-13
}node [pos=0.5,anchor=south,font=\footnotesize,sloped, text = red] {\textsf{BDD}};

\addplot [
  color=black,
  line width=1pt,
  mark=none
]
table {
6.000000  3.432204e-03
6.100000  3.122284e-03
6.200000  2.834637e-03
6.300000  2.568195e-03
6.400000  2.321902e-03
6.500000  2.094711e-03
6.600000  1.885590e-03
6.700000  1.693525e-03
6.800000  1.517522e-03
6.900000  1.356607e-03
7.000000  1.209834e-03
7.100000  1.076282e-03
7.200000  9.550588e-04
7.300000  8.453044e-04
7.400000  7.461898e-04
7.500000  6.569199e-04
}node [pos=0.5,anchor=south,font=\footnotesize,sloped,yshift=-0.1cm] {uncoded};
\end{axis}

\end{tikzpicture}
    \caption{FER versus crossover probability for eight  RS$[544, 514, 15]$ outer code interleaved with $64$ BCH$[700, 680, 2]$ inner code. The EaED$^{\mathsf{a}}$ uses the parameters $\te = 0.05$ and $\ta=0.56$ which minimizes the FER at $7$ dB.}
    \label{fig:RSBCH}
\end{figure}
 
\section{Application in Concatenated BCH-RS Schemes}
In~\cite{SukmadjiRSBCH}, a concatenated scheme with \ac{RS} outer codes and \ac{BCH} inner codes is presented and thoroughly analyzed. In this section, we show that only replacing the BDD used for the BCH codes in that scheme with EaED$^{\mathsf{a}}$ lowers the post-\ac{FEC} \ac{FER} with only a small increase in complexity. The rest of the scheme follows~\cite{SukmadjiRSBCH}. 

The scheme in~\cite{SukmadjiRSBCH} uses shortened BCH codes to reduce miscorrections and match the target code rate. 
Let ${\mathcal{H}_\mathsf{s}=[n_{\text{short}]}}$ denote the set of shortened positions.
Then, the subcode
$$\mathcal{C}^n_{\text{short}}:=\{\bm{c}\in\mathcal{C}[n,k]:c_i=0,\forall i\in \mathcal{H}_\mathsf{s}\}\subseteq\mathcal{C}[n,k]$$
is isomorphic to the shortened code $\mathcal{C}_{\text{short}}[\,n - n_{\text{short}},\, k - n_{\text{short}}]$ obtained by neglecting positions in $\mathcal{H}_\mathsf{s}$. 
Note that shortened bits are not transmitted, such that $\mathcal{C}_{\text{short}}$ is used for transmission and $\mathcal{H}_{\mathsf{a}}\cap\mathcal{H}_s=\emptyset$. We further assume that the error locations that the BDD declares are randomly distributed, and therefore, the rate that such an error location is in $\mathcal{H}_\mathsf{s}$ is $n_{\text{short}}/n$, which provides a good approximation when $n_{\text{short}}$ is small~\cite{SukmadjiRSBCH}. 
Therefore, we treat the shortened bits as additional correct anchor bits after decoding on $\mathcal{C}$.
We can then compute the \ac{DTP} of EaED$^{\mathsf{a}}$ using the method derived in Appendix~\ref{appendx:EaEDa}. 

Fig.~\ref{fig:RSBCH} shows the computed and simulated post-\ac{FEC} \ac{FER} using the framework in~\cite{SukmadjiRSBCH} and the \ac{DTP} of the \ac{BDD} and \ac{EaED}$^{\mathsf{a}}$.
\Ac{EaED}$^{\mathsf{a}}$ shows a clear decoding gain compared to \ac{BDD}. Regarding the computational complexity, \ac{EaED}$^{\mathsf{a}}$ involves at most two \ac{BDD} decodings.
When constructing the two test patterns, since transmitted codewords are randomly chosen, no random number generator is required, i.e., we can simply set ${\pone=\bm{0}}$ and ${\ptwo=\bm{1}}$. 
After decoding, \ac{MD} requires checking whether bits are anchor bits and, if two candidate codewords exist, comparing their Hamming distances. For memory complexity, additional storage is needed to record the positions of anchor and erasure bits. Overall, the increase in complexity is moderate, since the decoding of the \ac{RS} code, which has comparably high complexity, remains unchanged.

\section{Conclusion}
In this paper, we derive a closed-form expression for estimating the decoding behavior of EaED. Furthermore, results are extended to incorporate an enhanced version of the decoder. The MATLAB scripts used to generate the results are available at \url{https://github.com/kit-cel/EaED_BCH}.

\appendices
\section{EaED Calculation}\label{appendx:EaED}

\begin{table}[t]
    \centering
    \caption{All possible BDD outcomes when decoding the two test patterns in EaED, assuming that $2e+u\geq \dmin$.}
    \begin{tabular}{c c  c |  c c c c c | c c}
    \toprule
         & \ding{192}&\ding{193}&\ding{194}&\ding{195}&\ding{196}&\ding{197}&\ding{198}& \ding{199}&\ding{200} \\
         \hline
         $\BDD{\yone}$& $\bm{0}$ & $\bm{0}$ & $\mathsf{fail}$&$\bm{c}$&$\mathsf{fail}$ & $\bm{c}$&$\bm{c}$& $\bm{c}$&$\mathsf{fail}$ \\
         $\BDD{\ytwo}$& $\bm{c}$&$\mathsf{fail}$ &  $\mathsf{fail}$ & $\mathsf{fail}$& $\bm{c}$ & $\bm{c}$&$\bm{c}'$&$\bm{0}$ & $\bm{0}$\\
         \bottomrule
    \end{tabular}
    \label{tab:BDDoutcomes}
\end{table}

\begin{figure}[t]
    \centering
    \begin{tikzpicture}[scale=0.4]
\draw[-latex, very thick](0,2) -- (21,2);
\draw[-, very thick](5,1.8) -- (5,2.2);
\draw[-, very thick](15,1.8) -- (15,2.2);
\draw[-, very thick](0,1.8) -- (0,2.2);
\draw[-, very thick](20,1.8) -- (20,2.2);

\node at (0,1) {$0$};
\node at (20,1) {$e$};
\node at (5,1) {$\max(t-u,0)$};
\node at (15,1) {$\min(u+e-t,e)$};

\node at (21,3) {$e_1$};

\node at (2.5,3) {$\mathcal{L}:=\{$\ding{192}\ding{193}$\}$};
\node at (10,3) {$\mathcal{M}:=\{$\ding{194}\ding{195}\ding{196}\ding{197}\ding{198}$\}$};
\node at (17.5,3) {$\mathcal{R}:=\{$\ding{199}\ding{200}$\}$};

\end{tikzpicture}
    \caption{Sorting of the cases in Tab.~\ref{tab:BDDoutcomes} according to the value of $e_1$. When $u>t$, sets $\mathcal{L}$ and $\mathcal{R}$ will not occur.}
    \label{fig:EaEDcases}
\end{figure}

\begin{figure}[t]
    \centering
    \begin{tikzpicture}[scale=0.4]

\filldraw[draw=black, fill=white, thick] (0,1) rectangle (20,2);       %
\filldraw[draw=black, fill=white, thick] (0,3.5) rectangle (20,4.5);   %
\filldraw[draw=black, fill=white, thick] (0,6) rectangle (20,7);       %
\filldraw[draw=black, fill=white, thick] (0,8.5) rectangle (20,9.5);   %

\filldraw[draw=black, fill=orange!20] (12,3.5) rectangle (13,4.5);
\filldraw[draw=black, fill=orange!20] (19,3.5) rectangle (20,4.5);
\filldraw[draw=black, fill=green!20] (14,3.5) rectangle (15,4.5);
\filldraw[draw=black, fill=green!20] (17,3.5) rectangle (18,4.5);
\filldraw[draw=black, fill=gray!20] (13,3.5) rectangle (14,4.5);
\filldraw[draw=black, fill=gray!20] (18,3.5) rectangle (19,4.5);

\node[anchor=west] at (20.5,9.0) {$\bm{c}$};
\draw[decorate,decoration={brace,amplitude=2pt}]
(15,9.5) -- (20,9.5)
node[midway,yshift=8pt]{$r\!=\!u\!+\!e\!-\!e_1\!+\!a\!-\!b$};
\draw (15,8.5) -- (15,9.5);
\node at (7.5,9.0) {$0$};
\node at (17.5,9.0) {$1$};

\node[anchor=west] at (20.5,6.5) {$\ytwo$};
\draw (15,6) -- (15,7);
\draw (13,6) -- (13,7);
\draw (18,6) -- (18,7);
\draw[decorate,decoration={brace,amplitude=2pt}]
(13,7) -- (15,7)
node[midway,yshift=8pt]{$b$};
\draw[decorate,decoration={brace,amplitude=2pt}]
(18,7) -- (20,7)
node[midway,yshift=8pt]{$a$};
\draw[decorate,decoration={brace,amplitude=2pt}]
(0,7) -- (13,7)
node[midway,yshift=8pt]{$n-\wc-b$};
\draw[decorate,decoration={brace,amplitude=2pt}]
(15,7) -- (18,7)
node[midway,yshift=8pt]{$\wc - a$};
\node at (6.5,6.5) {$0$};
\node at (14,6.5) {$1$};
\node at (16.5,6.5) {$1$};
\node at (19,6.5) {$0$};

\node[anchor=west] at (20.5,4.0) {$\bm{y}$};
\draw (15,3.5) -- (15,4.5);
\draw (13,3.5) -- (13,4.5);
\draw (18,3.5) -- (18,4.5);
\draw (14,3.5) -- (14,4.5);
\draw (12,3.5) -- (12,4.5);
\draw (19,3.5) -- (19,4.5);
\draw (17,3.5) -- (17,4.5);
\node at (6,4.0) {$0$};
\node at (12.5,4.0) {$\que$};
\node at (13.5,4.0) {$1$};
\node at (14.5,4.0) {$\que$};
\node at (16,4.0) {$1$};
\node at (18.5,4.0) {$0$};
\node at (19.5,4.0) {$\que$};
\node at (17.5,4.0) {$\que$};
\draw[decorate,decoration={brace,amplitude=2pt}]
(17,4.5) -- (18,4.5)
node[midway,yshift=8pt]{$e-e_1-\gamma$};

\node[anchor=west] at (20.5,1.5) {$\yone$};
\draw (15,1) -- (15,2);
\draw (13,1) -- (13,2);
\draw (18,1) -- (18,2);
\draw (14,1) -- (14,2);
\draw (12,1) -- (12,2);
\draw (19,1) -- (19,2);
\draw (17,1) -- (17,2);
\node at (6,1.5) {$0$};
\node at (12.5,1.5) {$1$};
\node at (13.5,1.5) {$1$};
\node at (14.5,1.5) {$0$};
\node at (16,1.5) {$1$};
\node at (18.5,1.5) {$0$};
\node at (19.5,1.5) {$1$};
\node at (17.5,1.5) {$0$};
\draw[decorate,decoration={brace,amplitude=2pt}]
(12,2.0) -- (13,2.0)
node[midway,yshift=8pt]{$\lambda$};
\draw[decorate,decoration={brace,amplitude=2pt}]
(14,2.0) -- (15,2.0)
node[midway,yshift=8pt]{$\gamma$};
\draw[decorate,decoration={brace,amplitude=2pt}]
(19,2) -- (20,2)
node[midway,yshift=8pt]{$e_1-\lambda$};

\end{tikzpicture}
    \vspace{-2ex}
    \caption{Graphical illustration of case \ding{193}. Orange parts sum up to $e_1$ while green parts sum up to $e_2=e-e_1$. The gray parts sum up to $\dnE{\bm{y}}(\bm{c})$.}
    \vspace{-2ex}
    \label{fig:EaEDbarfigure}
\end{figure}
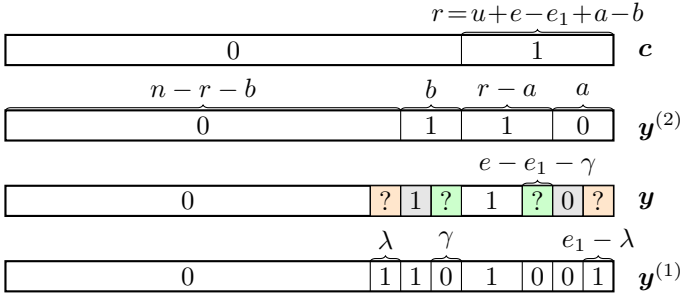
In this appendix, we derive \ac{DTP}
\(
P_\mathsf{EaED}(r|u,e)
\)  of \ac{EaED}. 

\textbf{Case 1}: If $e=0$, EaED reverts to BDD. Therefore, the transition probabilities are given by Theorem~\ref{theorem:BDD}.

\textbf{Case 2}: Consider $e\neq 0$. 

\textbf{Case 2.1}: If $2u+e<\dmin$, $\PEaEDsucc(u,e)=1$~\cite{rapp2021error}.

\textbf{Case 2.2}: When  ${2u+e\geq\dmin}$, EaED may still succeed.
Hence, we must calculate ${\PEaEDmc(r|u,e)}$, ${\PEaEDfail(r|u,e)}$, and ${\PEaEDsucc(r|u,e)}$ separately for this case.

Let the \ac{RV} $E_1=\weight{\pone}$; under the all-zero codeword assumption, $E_1=e_1$ corresponds to erroneous bits introduced by the random filling pattern to $\yone$ before \ac{BDD}, i.e.,
$\weight{\yone} = u+e_1$.
Analogously, define $e_2=\weight{\ptwo}$.
Note that ${E_1 + E_2 = E}$. Consequently, ${\weight{\ytwo} = u+e-e_1}$.

Table~\ref{tab:BDDoutcomes} lists all nine possible joint BDD outcomes.
In particular, note that ${\BDD{\yone}=\BDD{\ytwo}=\bm{0}}$ cannot occur, since this would imply $u+e_1\leq t$ and $u+e-e_1\leq t$, which together yield $2u+e\leq 2t \leq \dmin -1$, contradicting $u+2e\geq\dmin$.
Therefore, at most one of the two \ac{BDD} can correctly decode to $\bm{0}$.

As depicted in Fig.~\ref{fig:EaEDcases}, based on the value of $e_1$, the nine cases can be grouped into three categories (sets), denoted $\mathcal{L}:=\{$\ding{192}, \ding{193}$\}$, $\mathcal{M}=\{$\ding{194}, \ding{195}, \ding{196}, \ding{197}, \ding{198}$\}$, and $\mathcal{R}=\{$\ding{199}, \ding{200}$\}$, corresponding to the cases $e_1\in [0,t-u]$, $e_1\in (t-u,u+e-t)$ and $e_1\in [u+e-t, e]$, respectively. 
Then, we can compute 
\begin{align}    \label{eq:step1}
    &\PEaEDmc(r|u,e)\nonumber
    =\textstyle \sum_{e_1=0}^{e}\PEaEDmc(r|u,e,e_1)\cdot P(E_1=e_1|u,e)\nonumber\\
    \! &= \underbrace{\!2\!\sum_{e_1=0}^{t-u} \PEaEDmcL(r|u,e,e_1)  \frac{\binom{e}{e_1}}{2^e}}_{\mathrm{scenarios }\,\mathcal{L,R}}+\!\!\!\!\! \underbrace{\sum_{e_1=t-u+1}^{u+e-t-1}\!\!\!\!\PEaEDmcM(r|u,e,e_1) \frac{\binom{e}{e_1}}{2^e}}_{\mathrm{scenario}\,\mathcal{M}}.
\end{align}
The factor $2$ before the first sum is due to the symmetry between $\mathcal{L}$ and $\mathcal{R}$. 
Next, we compute the probabilities associated with $\mathcal{L}$ and $\mathcal{M}$.

The set $\mathcal{L}$ contains the joint decoding outputs for which $e_1\leq t-u$, which implies that $\BDD{\yone}=\bm{0}$.
Furthermore, $\weight{\ytwo}=u+e-e_1\geq 2u+e-t\geq \dmin-t >t$.
Therefore, $\BDD{\ytwo}$ is either a miscorrection (case \ding{192}) or decoding failure (case \ding{193}).

There are ${\binom{n}{u}\binom{n-u}{e}\binom{e}{e_1}:=\Theta}$ possible combinations of locations for the $u$ errors, $e$ erasures, and $e_1$ $1$s among the $e$ erased positions when generating the test pattern $\yone$. Then, we first count the number of events leading to \ding{192}, from which \ding{193} follows directly.

Now, for case \ding{192}, we fix a non-zero codeword $\bm{c}$ and first enumerate all possible $\ytwo$. Based on each $\ytwo$, we then infer and count all elementary events $\bm{y}$ consistent with the realization $(u,e,e_1)$.

We now first count the number of $\ytwo$ patterns with weight $u + e_2$ that fulfill $ \ytwo \in \mathcal{S}_t(\bm{c})$ for some ${\bm{c}\in\mathcal{C}\setminus\{\bm{0}\}}$, with $\weight{\bm{c}}=r$.
To this end, as illustrated in Fig.~\ref{fig:EaEDbarfigure} and similarly to \ac{BDD}, we first fix $\ytwo$ and introduce variables ${a = |\{i\in[n]:c_i=1,y^{(2)}_i = 0\}|}$ and ${b = |\{i\in[n]:c_i=0,y^{(2)}_i =1\}|}$. As $\BDD{\ytwo}=\bm{c}$, $(a,b)\in \mathcal{S}_{a,b}=\{(a,b)\in [t]\times  [t]: a+b\leq t, u+e\!-\!e_1+a-b\!=\!r\}$. There are $A_{r}\binom{r}{a}\binom{n-r}{b}$ possible such $\ytwo$ patterns.

Then, to obtain the number of possible $\bm{y}$ patterns with $\ytwo$ fixed, we distribute $e$ erasures across the four different parts of $\ytwo$, which are divided according to the four possible values of $(c_i,y^{(2)}_i)$.
Next, we introduce the variables ${\lambda:=|\{i\in[n]:c_i=0,\, y_i = \que, \,y^{(1)}_i = 1\}|}$ and ${\gamma:=|\{i\in[n]:c_i=0,\, y_i = \que, \,y^{(1)}_i = 0\}}$.
For a fixed pair of $\bm{c}$ and $\ytwo$, there are $\binom{n-r-b}{\lambda}\binom{b}{\gamma}\binom{\wc-a}{e-e_1-\gamma}\binom{a}{e_1-\lambda}$ distinct sequences $\bm{y}$ when considering the possible locations of the erasures.
Observe that whenever $\ytwo$ and $\bm{y}$ are specified, the word $\yone$ is uniquely determined.
Hence, the number of placements for case \ding{192} amounts to 
\begin{equation}
\label{eq:eaedcase2}
    \textstyle\sum_{(a,b)\in \mathcal{S}_{a,b}}\sum_{\lambda\in\mathcal{S}_{\lambda}}\sum_{\gamma\in \mathcal{S}_{\gamma}} V(a,b,\lambda,\gamma)
\end{equation}
where $\mathcal{S}_{\lambda}=\{\lambda\in [n-r-b]:0\leq e_1-\lambda\leq a\}$, $\mathcal{S}_{\gamma}=\{\gamma\in [b]:0\leq e-e_1-\gamma \leq r-a\}$, and $$V(a,b,\lambda,\gamma) = A_{r}\tbinom{r}{a}\tbinom{n-r}{b}\tbinom{n-r-b}{\lambda}\tbinom{b}{\gamma}\tbinom{\wc-a}{e-e_1-\gamma}\tbinom{a}{e_1-\lambda},$$which is abbreviated as $V$ when unambiguous.

Next, for case \ding{192}, we show that evaluating \eqref{eq:eaedcase2} in a refined range of $\lambda$ and $\gamma$ yields the number of elementary cases that lead to different EaED outcomes.

As depicted in Fig.~\ref{fig:EaEDbarfigure}, $\dnE{\bm{y}}(\bm{c}) = b-\gamma+a-(e_1-\lambda)$ and obvisouly $\dnE{\bm{y}}(\bm{0})=u$.
Hence, if ${\dnE{\bm{y}}(\bm{c})<\dnE{\bm{y}}(\bm{0})}$, a miscorrection occurs as \ac{EaED} chooses $\bm{c}$ as its estimate. Similarly, when ${\dnE{\bm{y}}(\bm{c})}>\dnE{\bm{y}}(\bm{0})$, \ac{EaED} chooses $\bm{0}$ as its output and will decode successfully. Finally, When ${\dnE{\bm{y}}(\bm{c})=\dnE{\bm{y}}(\bm{0})}$, miscorrection and success happen with equal probability. 

For convenience, let 
$\mathcal{S}_{\odot}:=\{(\lambda,\gamma):\lambda-\gamma \odot u+e_1-a-b\}$, with $\odot \in \{>,<,=,\leq,\geq\}$ denoting a comparator. Thereby, let
\begin{align*}
    &L_{11}(r) := \textstyle\sum_{(a,b)\in \mathcal{S}_{a,b}}\sum_{(\lambda,\gamma)\in (S_{\lambda}\times \mathcal{S}_{\gamma})\cap \mathcal{S}_{<}}V,\\
    &L_{12}(r) := \textstyle\sum_{(a,b)\in \mathcal{S}_{a,b}}\sum_{(\lambda,\gamma)\in (S_{\lambda}\times \mathcal{S}_{\gamma})\cap \mathcal{S}_{>}}V,\\
    &L_{13}(r) :=\textstyle \sum_{(a,b)\in \mathcal{S}_{a,b}}\sum_{(\lambda,\gamma)\in (S_{\lambda}\times \mathcal{S}_{\gamma})\cap \mathcal{S}_{=}}V.
\end{align*}
Then, the number of elementary events that lead to case \ding{193} is
\begin{equation}
    \label{eq:L1}
    L_2 = \Theta - (L_{11} + L_{12} + L_{13})
\end{equation}
where $L_i = \sum_{r=u+e_1-t}^{u+e_1+t}L_i(r)$ for $i\in \{a,b,c\}$.

In summary, for the cases in $\mathcal{L}$, we have \begin{equation}
    \label{eq:probinL}
    \begin{cases}
        \PEaEDmcL(R=r|u,e,e_1) = (L_{11}(r) + L_{13}(r)/2)/\Theta,\\
    \PEaEDfailL(R=u+\frac{e}{2}|u,e,e_1)=0,\\
    \PEaEDsuccL(R=0|u,e,e_1)= (L_2 + L_{12} + L_{13}/2)/\Theta.
    \end{cases}
\end{equation}

Finally, we consider the cases in $\mathcal{M}$. First, note that they can only cause a miscorrection or a failure since $\yone,\ytwo\notin\mathcal{S}_t({\bm{0}})$. 
In the following, we make a significant simplification by assuming that the outcomes of $\BDD{\yone}$ and $\BDD{\ytwo}$, whether miscorrection or failure, are independent of each other.
In principle, this assumption is false as $d(\yone,\ytwo)=e$. 
However, as the positions of the erasures are random, we observe that fixing the decoding result of one BDD does not significantly influence the other one. Thus,
\begin{equation}\label{eq:caseMfail}
    \PEaEDfailM(u,e,e_1) = \PBDDfail(u\!+\!e_1)\PBDDfail(u\!+\!e\!-\!e_1).
\end{equation}
As for miscorrections, we assume that when both $\yone$ and $\ytwo$ are miscorrected, the probability of choosing one of the miscorrected codewords is $\frac{1}{2}$. Then we can approximate
\begin{align}\label{eq:caseM}
\PEaEDmcM(r&|u,e,e_1) = \PBDDmc(r|u\!+\!e_1)\PBDDfail(u\!+\!e\!-\!e_1) \notag\\
&\quad+ \PBDDfail(u\!+\!e_1)\PBDDmc(r|u\!+\!e\!-\!e_1) \notag\\
&\quad+ \tfrac{1}{2}\PBDDmc(r|u\!+\!e_1)\PBDDmc(u\!+\!e\!-\!e_1) \notag\\
&\quad+ \tfrac{1}{2}\PBDDmc(r|u\!+\!e\!-\!e_1)\PBDDmc(u\!+\!e_1).
\end{align}

Now we have obtained the probabilities for cases in $\mathcal{L}$ and $\mathcal{M}$. Plugging the results into \eqref{eq:step1} and marginalizing over the decoding outcomes $\mathcal{D}$ yields Theorem~\ref{theorem:EaED}.

\section{EaED with Anchor Bits}\label{appendx:EaEDa}

Based on the calculation in Appendix~\ref{appendx:EaED}, the \ac{DTP} of the EaED$^{\mathsf{a}}$ can be obtained by considering that codewords found by $\BDD{\yone}$ and $\BDD{\ytwo}$ could be rejected by the \ac{MD}. 
To this end, let $\bm{z}$ denote one of $\yone$ or $\ytwo$ and ${\BDD{\bm{z}} = \bm{c}}$. 
For ${\bm{c}\in \mathcal{C}}$, define ${\nca:=\{i\in [n] : z_i=0,c_i=1\}}$ and ${\nwa:=\{i\in [n] : z_i=1,c_i=0\}}$.
Then, the probability that $\bm{c}$ is accepted by the \ac{MD} is $(1-\pca)^{\nca}(1-\pwa)^{\nwa}=:p_{\text{AC}}$, since the events of individual bits being anchor bits are statistically independent. Before we derive the $\PEaEDa$, we first compute the \ac{DTP} of decoding one random test pattern, say, $\yone$ when $u+e_1>t$. It is denoted as $\PEaEDsa$.
We define  \[ L_s(r)=\sum_{a,b,\nca,\nwa}A_{r}\tbinom{r}{a}\tbinom{n-r}{b}\tbinom{b}{\nwa}\tbinom{a}{\nca}\tbinom{r-a}{c}\tbinom{n-r-b}{d}\cdot p_{\text{AC}}\] where $c = e_1-(b-\nwa)$ and $d = e-e_1-(a-\nca)$. Then, the miscorrection probability is
\begin{equation}\label{eq:singlepattern}
    \PEaEDsamc(r|u,e,e_1) = L_s(r)/\Theta,
\end{equation}
and $\PEaEDsafail(u,e,e_1) = 1-\PEaEDsamc(u,e,e_1)$ as the success probability is $0$ when $u+e_1>t$.

Next, we analyze the decoding result by cases.

\textbf{Case 1}: For $e=0$, EaED$^{\mathsf{a}}$ reverts to BDD$^{\mathsf{a}}$.

\textbf{Case 2}: For $e\neq 0$, we need to compute two cases. 

\textbf{Case 2.1}: If {$2u+e<\dmin$}, EaED is guaranteed to succeed. However, due to the \ac{MD} step, there are three sub-cases. 

\textbf{Case 2.1.1}: If ${u\!+\!e_1\leq t}$ and ${u\!+\!e\!-\!e_1<\dmin\!-\!t}$ or if ${u\!+\!e_1<\dmin\!-\!t}$ and ${u\!+\!e\!-\!e_1\leq t}$, either $\yone$ or $\ytwo$ will be decoded to $\bm{c}=\bm{0}$ (with $\nwa=u$ and $\nca=0$) and the other one results in a decoding failure.
Let $P_{\bm{0}} :=  (1 - \pwa)^{u}$.
Then \({\PEaEDasucc(u,e) = P_{\bm{0}}}\) and \( {\PEaEDafail(u,e) = 1-P_{\bm{0}}}.\) 

\textbf{Case 2.1.2}: If \(u + e_{1} \leq t\) and \(u + e - e_{1} \geq d_{\min} - t\). ${\BDD{\yone}\!=\!\bm{0}}$, and $\BDD{\ytwo}$ either fails or miscorrects. For a miscorrection where $\BDD{\ytwo}=\bm{c}\neq \bm{0}$ with $\weight{\bm{c}}=r$, $\dnE{\bm{y}}(\bm{0})< \dnE{\bm{y}}(\bm{c})$ always holds, as shown in~\cite[Appendix~A]{rapp2021error}. Therefore, if $\bm{0}$ is accepted by the \ac{MD} (probability $P_{\bm{0}}$), it will be chosen as the EaED decision. Then, ${\PEaEDasucc(u,e)=P_{\bm{0}}}$. When $\bm{0}$ is rejected (probability $(1-P_{\bm{0}})$), EaED will either fail or miscorrect. We compute $\PEaEDamc(r|u,e)\! =\! (1-P_{\bm{0}})\PEaEDsamc\left(r|u,e,e_2\right)$ and
$\PEaEDafail(u,e) = 1-\PEaEDasucc(u,e)-\PEaEDamc(u,e)$.

\textbf{Case 2.1.3}: If \(u + e - e_{1} \leq t\) and \(u + e_{1} \geq d_{\min} - t\),  the probabilities are calculated with the same formulas as in case~2.1.2, by redefining $d_{\bm{y}\bm{c}} = d_{\min} - (u + e_{1})$.

The case where both $\yone$ and $\ytwo$ are miscorrected is impossible as it contradicts the condition $2u+e<\dmin.$ Next, we consider the case where $2u+e\geq \dmin$.

\textbf{Case 2.2}: $2u+e\geq \dmin$. Analogous to Appendix~\ref{appendx:EaED}, we divide the cases into sets $\mathcal{L}$, $\mathcal{M}$, and $\mathcal{R}$, and observe that we can replace the probabilities associated with EaED in \eqref{eq:step1} by their EaED$^{\mathsf{a}}$ counterparts.

For set $\mathcal{L}$, as depicted in Fig.~\ref{fig:EaEDbarfigure}, we write $P_{\bm{0}}=(1-\pwa)^u$ and $P_{\bm{c}}=(1-\pca)^{a - (e_1-\lambda)}  (1-\pwa)^{b - \gamma}$. We can reuse \eqref{eq:probinL}, replacing $L_{11}$, $L_{12}$, and $L_{13}$ by $L_{11}^{\mathsf{a}}$, $L_{12}^{\mathsf{a}}$, and $L_{13}^{\mathsf{a}}$, respectively, which are calculated as
\begin{align*}
&L_{11}^{\mathsf{a}}=\sum_{a,b}\bigg(\sum_{(\lambda,\gamma)\in (S_{\lambda}\times \mathcal{S}_{\gamma})\cap \mathcal{S}_{<}}\!\!\!\!\!\!\!\!\!\!\!\!V \cdot P_{\bm{c}} + \sum_{(\lambda,\gamma)\in (S_{\lambda}\times \mathcal{S}_{\gamma})\cap \mathcal{S}_{\geq}}\!\!\!\!\!\!\!\!\!\!\!\!\!\!V\cdot (1-P_{\bm{0}})P_{\bm{c}}\bigg)\\
&L_{12}^{\mathsf{a}} = \sum_{a,b}\bigg(\sum_{(\lambda,\gamma)\in (S_{\lambda}\times \mathcal{S}_{\gamma})\cap \mathcal{S}_{\leq}}\!\!\!\!\!\!\!\!\!\!\!\!\! V \cdot P_{\bm{0}}(1-P_{\bm{c}})+\sum_{(\lambda,\gamma)\in (S_{\lambda}\times \mathcal{S}_{\gamma})\cap \mathcal{S}_{>}}\!\!\!\!\!\!\!\!\!\!\!\!\!V \cdot P_{\bm{0}}\bigg)\\
&L_{13}^{\mathsf{a}}=\sum_{a,b}\sum_{(\lambda,\gamma)\in (S_{\lambda}\times \mathcal{S}_{\gamma})\cap \mathcal{S}_{=}}\!\!\!\!\!\!\!\!\!\!\!\!\!\!V\cdot  P_{\bm{0}}P_{\bm{c}}
\end{align*}
where $(a,b)\in \mathcal{S}_{a,b}$.
Additionally, define
\[L_{\mathsf{fail}}=\textstyle\sum_{(a,b)\in \mathcal{S}_{a,b}}\sum_{(\lambda,\gamma)\in (S_{\lambda}\times \mathcal{S}_{\gamma})}V(1-P_{\bm{c}})(1-P_{\bm{0}}).\]

Then, we have \begin{equation}
    \label{eq:probinLEaEDa}
    \begin{cases}
        \PEaEDamcL(R=r|u,e,e_1) \!= \!(L_{11}^{\mathsf{a}}(r) + L^{\mathsf{a}}_{13}(r)/2)/\Theta,\\
    \PEaEDafailL(u,e,e_1)=(L_{\mathsf{fail}}+L_2(1-P_{\bm{0}}))/\Theta,\\
    \PEaEDasuccL(u,e,e_1)\!= \!(L_2P_{\bm{0}} + L^{\mathsf{a}}_{12} + L^{\mathsf{a}}_{13}/2)/\Theta.
    \end{cases}
\end{equation}

Finally, we consider the set $\mathcal{M}$, where we still assume that the decoding results of $\yone$ and $\ytwo$ are independent.
 Then, the probabilities are obtained in the same way as for EaED in \eqref{eq:caseM}, except that the BDD probabilities are replaced with those of a single test-pattern EaED computed by~\eqref{eq:singlepattern}. Finally, the DTP of EaED$^{\mathsf{a}}$ can be computed using the framework of Theorem~\ref{theorem:EaED}.

\section*{Acknowledgment}
This work has received funding from the European Research Council (ERC) under the European Union’s Horizon 2020 research and innovation programme (grant agreement No. 101001899).

\bibliographystyle{IEEEtran}
\bibliography{bib}

\end{document}